\documentclass[superscriptaddress,10pt,aps,pra,twocolumn,longbibliography, notitlepage]{revtex4-2}
\usepackage{float}
\floatstyle{boxed}
\usepackage{array}
\usepackage{graphicx}
\usepackage{amsbsy}
\usepackage[utf8x]{inputenc}
\usepackage{epstopdf}
\usepackage{amsmath,amssymb,amsfonts,amsthm}
\usepackage{indentfirst}
\usepackage{soul}
\usepackage[T1]{fontenc}
\usepackage[dvipsnames]{xcolor}
\usepackage{url}
\usepackage[colorlinks]{hyperref}

\hypersetup{
    plainpages=true,
    breaklinks=true,
    hypertexnames=false,
    pageanchor=true,
    colorlinks=true,
    linkcolor={blue},
    citecolor={red},
    urlcolor={blue},
    anchorcolor={black}
}

\newcommand{\figref}[1]{\mbox{Fig.~\ref{#1}}}

\renewcommand{\eqref}[1]{\mbox{Eq.~(\ref{#1})}}

\newcommand{\be}{\begin{equation}}
\newcommand{\ee}{\end{equation}}
\newcommand{\bea}{\begin{eqnarray}}
\newcommand{\eea}{\end{eqnarray}}

\newtheorem{theorem}{Theorem}

\begin{document}


\title{Fully solvable $n$-dimensional finite simplex lattices with a harmonic spectrum}

\title{Fully solvable non-Hermitian finite simplex lattices in arbitrary dimensions}

\title{Fully solvable  finite simplex lattices with open boundaries in arbitrary dimensions}



\author{Ievgen I. Arkhipov}
\email{ievgen.arkhipov@upol.cz} 
\affiliation{Joint Laboratory of
Optics of Palack\'y University and Institute of Physics of CAS,
Faculty of Science, Palack\'y University, 17. listopadu 12, 771 46
Olomouc, Czech Republic}

\author{Adam Miranowicz}
\affiliation{Theoretical
Quantum Physics Laboratory,  Cluster for Pioneering Research,
RIKEN, Wako-shi, Saitama 351-0198, Japan} \affiliation{Quantum
Information Physics Theory Research Team, Quantum Computing
Center, RIKEN, Wakoshi, Saitama, 351-0198, Japan}
\affiliation{Institute of Spintronics and Quantum Information,
Faculty of Physics, Adam Mickiewicz University, 61-614 Pozna\'n,
Poland}

\author{Franco Nori}
\affiliation{Theoretical Quantum Physics
Laboratory,  Cluster for Pioneering Research, RIKEN,  Wako-shi,
Saitama 351-0198, Japan} \affiliation{Quantum Information Physics
Theory Research Team, Quantum Computing Center, RIKEN, Wakoshi,
Saitama, 351-0198, Japan} \affiliation{Physics Department, The
University of Michigan, Ann Arbor, Michigan 48109-1040, USA}

\author{\c{S}ahin K. \"Ozdemir}
\affiliation{Department of Engineering
Science and Mechanics, and Materials Research Institute (MRI), The
Pennsylvania State University, University Park, Pennsylvania
16802, USA}

\author{Fabrizio Minganti}
\affiliation{ Institute of
Physics, Ecole Polytechnique F\'ed\'erale de Lausanne (EPFL),
CH-1015 Lausanne, Switzerland} \affiliation{Center for Quantum
Science and Engineering, Ecole Polytechnique Fédérale de
Lausanne (EPFL), CH-1015 Lausanne, Switzerland}

\begin{abstract}
Finite simplex lattice models are used in different branches of
science, e.g., in condensed matter physics,  when studying
frustrated magnetic systems and non-Hermitian localization
phenomena; or in chemistry, when describing experiments with
mixtures. An  $n$-simplex represents the simplest possible
polytope in $n$ dimensions, e.g., a line segment, a triangle, and
a tetrahedron in one, two, and three dimensions, respectively. In
this work, we show that various fully solvable, in general non-Hermitian, 
$n$-simplex lattice models {with open boundaries} can be constructed from the high-order field-moments
space of quadratic bosonic systems. Namely, we demonstrate
that such $n$-simplex lattices can be formed by a dimensional
reduction of highly-degenerate iterated polytope chains in
$(k>n)$-dimensions, which naturally emerge in the field-moments space. 
Our
findings indicate that the field-moments space of bosonic systems provides a
versatile platform for simulating real-space $n$-simplex
lattices exhibiting non-Hermitian phenomena, and yield valuable insights into the structure of
many-body systems exhibiting similar complexity.
{Amongst a variety of practical applications, these simplex structures can offer a physical setting for  implementing the discrete fractional Fourier transform, an indispensable tool for both quantum and classical signal processing.}
\end{abstract}

\date{\today}

\maketitle

\section{Introduction}

Simplex lattice models are used in various fields of science. In
physics, they can describe antiferromagnetic Ising
models~\cite{Aizemann1981,Wu2013,Baxter2020} and non-Hermitian
localization
phenomena~\cite{Hatano1996,Ashida2020,Bergholtz2021,Ding2022}. In
chemistry, simplex lattices describe blending processes in
multi-component systems~\cite{Scheffe1958}, and in biology, they
are used to model arrays of myosin filaments in higher vertebrate
muscle packs~\cite{LUTHER1980}. A simplex is the simplest possible
polytope in a given dimension, such as a line segment, triangle,
and tetrahedron in 1D, 2D, and 3D spaces, respectively. Polytopes
generalize polyhedra in arbitrary dimensions~\cite{GrunbaumBook}
and have also been widely used in condensed matter
physics~\cite{Nelson1984}.

Finding the full and exact set of eigenvalues and eigenvectors of
finite lattice models in any dimension is a difficult task. While
there has been significant progress in finding exact algebraic
solutions for certain classes of periodic
systems~\cite{BaxterBook,KunduBook,Kitaev2006,Yao2007,Nussinov2009},
as well as multi-dimensional Hermitian
lattices~\cite{Nussinov2009}, obtaining an analytical expression
for the eigenspectrum of finite simplex non-periodic, i.e., with open boundaries, systems
remains challenging~\cite{Leumer2020}. 

Finding new classes of
solvable models is also important for benchmarking advanced
numerical methods~\cite{Ors2019,Carleo2019}. In addition to their
mathematical value, the practical implementation of discovered
exactly solvable lattice models in experiments holds significant
importance, allowing to gain valuable insights into the nature and
behavior of seemingly complex physical systems.

Here we present an analytical solution for the eigenspectrum of finite $n$-simplex lattice models {with open boundaries} in {\it arbitrary} $(n\geq 1)$
dimensions, which, most importantly, can naturally emerge in the
field-moments (FMs) space of bosonic systems and, thus, can be
immediately simulated with state-of-art experimental setups.
{In this context, the term {\it open boundaries} implies the absence of periodic boundary conditions on such lattices.}

At
the heart of the proposed solution is the fact that these
$n$-simplex lattices can be generated by a dimensional reduction
of higher-dimensional iterated polytope chains (IPCs) into
low-dimensional iterated simplex chains (ISCs), which are formed
in the bosonic FMs space. This space folding can be considered as
a special projection of $k$-polytopes on an $n$-dimensional
hypersurface, with a subsequent relabeling of the vertex and link
weights of a formed lattice.  The procedure, thus, echoes the
cut-and-project method, used for obtaining quasicrystalline
structures from higher-dimensional regular
crystals~\cite{Kalugin1985,Elser1985,Katz1986,Nori1988,Baake1991},
and  extends it to non-Hermitian systems. The exact eigenspectrum
for the formed $n$-simplex lattices can then be readily attained
by exploiting the tensor-product-states character of the
eigenfunctions of IPCs.

The present study also generalizes the
eigendecomposition of certain tridiagonal matrices~\cite{Narducci1972,Hu2021}, used to describe, e.g., angular momentum of quantum light and a number of various photonic non-Hermitian 1D simplex lattices~\cite{Joglekar2011,arkhipov2023a}. 
The findings thus can pave the road to the study of nontrivial
non-Hermitian effects in high-dimensional simplex structures.

Our results have potential applications in quantum simulations of
non-Hermitian systems. The experimental observation of various
many-body phenomena in engineered lattices remains challenging due
to the difficulty in controlling all system
parameters~\cite{Daley2022}. Our proposed method offers a
promising solution to this problem, as it can be implemented using
the FMs space of simple bosonic systems~\cite{Nori1994}.
{The FMs space serves as a universal tool, which can provide a broad variety of physical quantities not only in bosonic but also in fermionic systems, ranging from magnetoconductance of electrons to superconducting-normal phase boundaries~\cite{Lin1996a,Lin1996b,Lin1996c,Lin2002}.
That is, the studied $n$-simplex lattices, emerging in the FMs space of quadratic systems, can emulate certain lattice Hamiltonians, allowing for a much simpler,  controllable, and scalable
quantum simulation of high-dimensional lattices in the synthetic
FM space of few-boson
systems~\cite{Regen2012,Yuan2018,Ozawa_2019,Lustig2019,Fabre2022,Weimann2017,Roy2021b,Weidemann2020,Teimourpour2014,Quiroz19,Tschernig20}.} 
The proposed procedure thus opens new avenues for the use of purely
Hamiltonian quantum systems, such as linear optical quantum
computers \cite{Kok2007}, for the simulation of nontrivial
non-Hermitian phenomena~\cite{Arkhipov2021a}. 

{Furthermore, as a technological application, we show that these simplex lattices, emerging in the FMs space, can also be exploited for implementing the discrete fractional Fourier transform (DFrFT). 
The DFrFT generalizes the ordinary discrete Fourier transform
in a close analogy as the continuous fractional Fourier transform
generalizes the continuous ordinary Fourier
transform~\cite{Candan2000}.
The DFrFT plays a central role in quantum and classical signal processing~\cite{Weimann2016,Zhou2017,Niewelt2023}. 
Namely, we demonstrate that  these lattice structures can be mapped to the operator space of quantum angular momentum, which is used for implementation of DFrFT~\cite{Weimann2016}.}

{This paper is organized as follows. In Sec.~II we elaborate the method for constructing the finite simplex lattices with harmonic spectrum, which naturally appear in the FMs space of quadratic systems. Then, we describe in detail the geometry of these simplex structures and present an exact solution for their eigenspace.  In Sec.~III, we discuss the potential and significance of the revealed $n$-simplex lattices, highlighting their relevance in both theoretical and practical applications. 
Conclusions are drawn in Sec.~IV.}

\section{The method}

To explain the idea of the proposed method for constructing fully solvable harmonic finite simplex lattices, we first consider the
specific example of quadratic coupled bosonic modes, where the FMs
space gives rise to the IPCs structure. That is, we first study
the composition of the evolution matrices governing the dynamics
of FMs of this quadratic Markovian system, and then show that the
$m$-series of evolution matrices, governing the $i$th-order FMs
($i=1,\dots,m$),  can be associated with the corresponding
$m$-chain of high-dimensional polytopes, whose highly-degenerate
eigenspace is formed by tensor-product-states. Following that, we briefly outline a
dimensional reduction procedure of the highly-degenerate space of
the IPCs, leading to the formation of exactly solvable ISCs that
can emulate real-space $n$-simplex lattices.

\subsection{Evolution matrices governing high-order field-moments
space}

Let us consider a quadratic, in general non-Hermitian, $N$-mode
system, whose Hamiltonian is  {
\begin{equation}
  \hat H = \sum\limits_{m,n} H_{mn}\hat \Psi_m^{\dagger}\hat \Psi_n
  \label{Hamiltonian}
\end{equation} }
with $\hat \Psi_j$ the $j$th element of the Nambu vector 
\begin{equation}
  \hat\Psi=\left[\hat a_1,\hat a_2,\dots,\hat a_N,\hat
a_1^{\dagger},\hat a_2^{\dagger},\dots,\hat
a_N^{\dagger}\right]^T,
  \label{Nambu}
\end{equation}
where $\hat a_k$ ($\hat a_k^{\dagger}$) is the annihilation
(creation) operator of the mode $k$, obeying $\left[\hat a_k,\hat
a_l^{\dagger}\right]=\delta_{kl}$ and $\left[\hat a_k,\hat
a_l\right]=0$. From the Heisenberg equations of motion,  one can
easily write down the equations for the dynamics of the
first-order field moments $\langle\hat \Psi\rangle$ as
\begin{equation}
  \dfrac{{\rm d}}{{\rm d}t}\langle\hat \Psi\rangle = M_1\langle\hat
  \Psi\rangle,
  \label{M_1}
\end{equation}
where $M_1$ is the corresponding evolution matrix for the
first-order FMs.  Note that the same procedure can be
straightforwardly extended to Hermitian operators as well
to, e.g.,  bosonic quadrature operators.

The analytical form of an evolution matrix governing the dynamics
of any higher-order FMs can be obtained by exploiting properties
of matrices formed by Kronecker sums~\cite{Arkhipov2021a}. Namely,
any $m$th order FM vector, constructed from the moments of the
tensor products of the Nambu vector $\hat\Psi$, i.e.,
${\Big\langle\bigotimes\limits_1^m\hat\Psi\Big\rangle}$, is
governed by the evolution matrix $M_m$, which is obtained from
$M_1$ as follows
\begin{equation}\label{M_m}
\frac{{\rm d}}{{\rm
d}t}{\Big\langle\bigotimes\limits_1^m\hat\Psi\Big\rangle}=M_m{\Big\langle\bigotimes\limits_1^m\hat\Psi\Big\rangle},
\quad \text{with} \quad M_m=\bigoplus\limits_1^m M_1,
\end{equation}
where symbols $\bigotimes$ and $\bigoplus$ stand for the Kronecker
tensor product and sum, respectively. For more details see
Appendix~\ref{AA} and Ref.~\cite{Lototsky2015}. For simplicity, in
\eqref{M_m}, we ignore any possible presence of inhomogeneous
parts stemming from quantum noise.  Note that for odd-order
non-Hermitian FMs, the inhomogeneous part is always absent for
weakly coupled Markovian systems, for even-order Hermitian FMs,
the noise might enter the r.h.s. of \eqref{M_m}, though it does
not affect the frequency spectrum of a system~\cite{Perina2022}.
\begin{figure}[t!]
    \centering
    \includegraphics[width=\columnwidth]{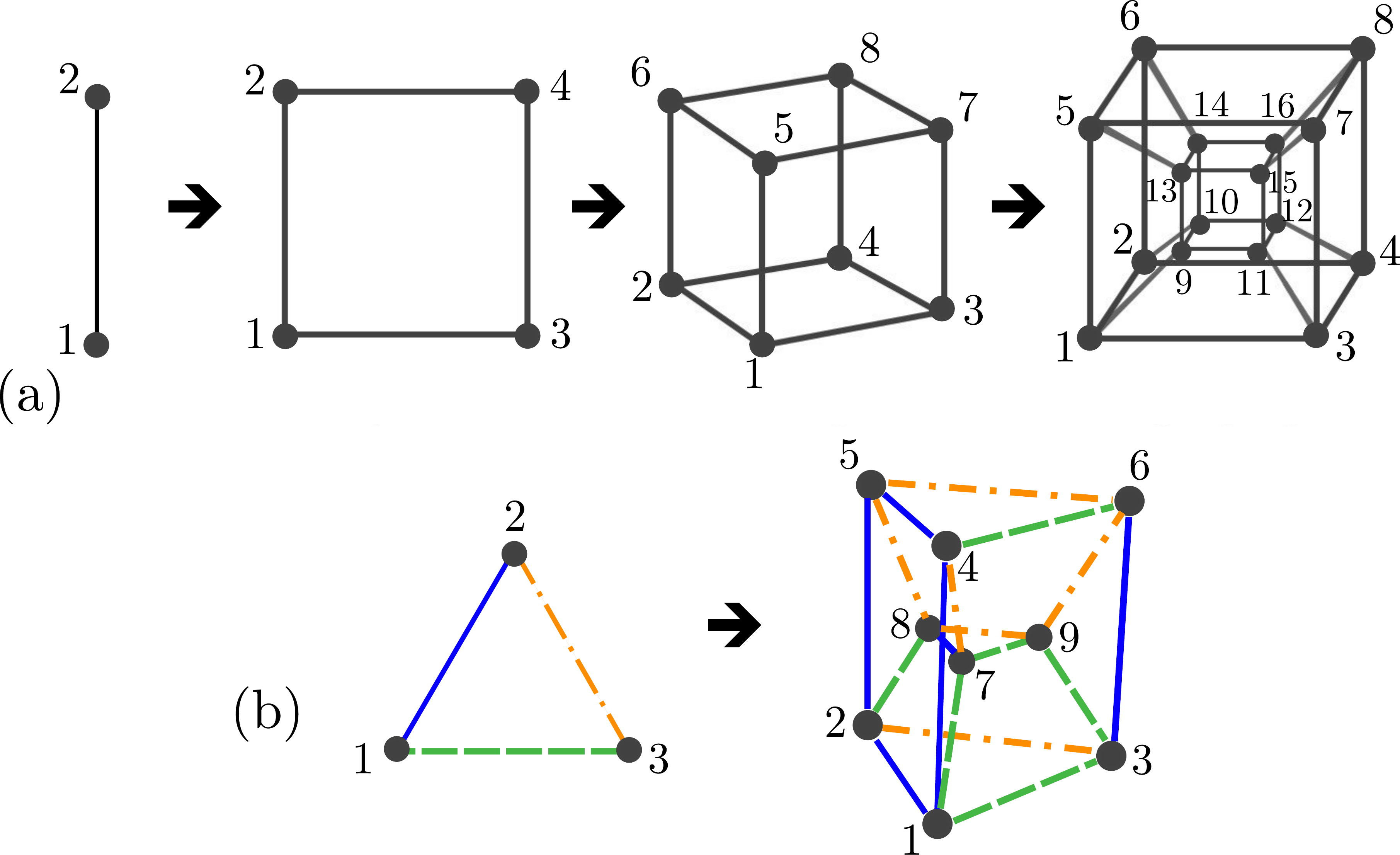}
\caption{Iterative polytope chains in the field-moments (FMs) space of bosonic (a) dimers and (b)
trimers, formed by the Cartesian products of 1-polytope (a line)
and 2-polytope (a triangle), respectively. (a) 1-polytope (a line)
for the first-order FMs of the dimer; 2-polytope (a square) for
its second-order FMs; 3-polytope (a cube) for its third-order FMs,
and 4-polytope (a hypercube, shown via its 3D Schlegel diagram)
for its fourth-order FMs. (b) 2-polytope (a triangle) for the
first-order FMs of the trimer, and 4-polytope (a 4D duoprism,
shown via its 3D Schlegel diagram) for its second-order FMs.
Numbered vertices correspond to the elements of the FMs vector
${\langle\bigotimes^k_1\hat\Psi\rangle}$, with vertex potentials
and the link weights corresponding to the diagonal and
off-diagonal elements of the related evolution matrix $M_k$,
respectively.}
    \label{fig1}
\end{figure}
\begin{figure}
    \centering
    \includegraphics[width=\columnwidth]{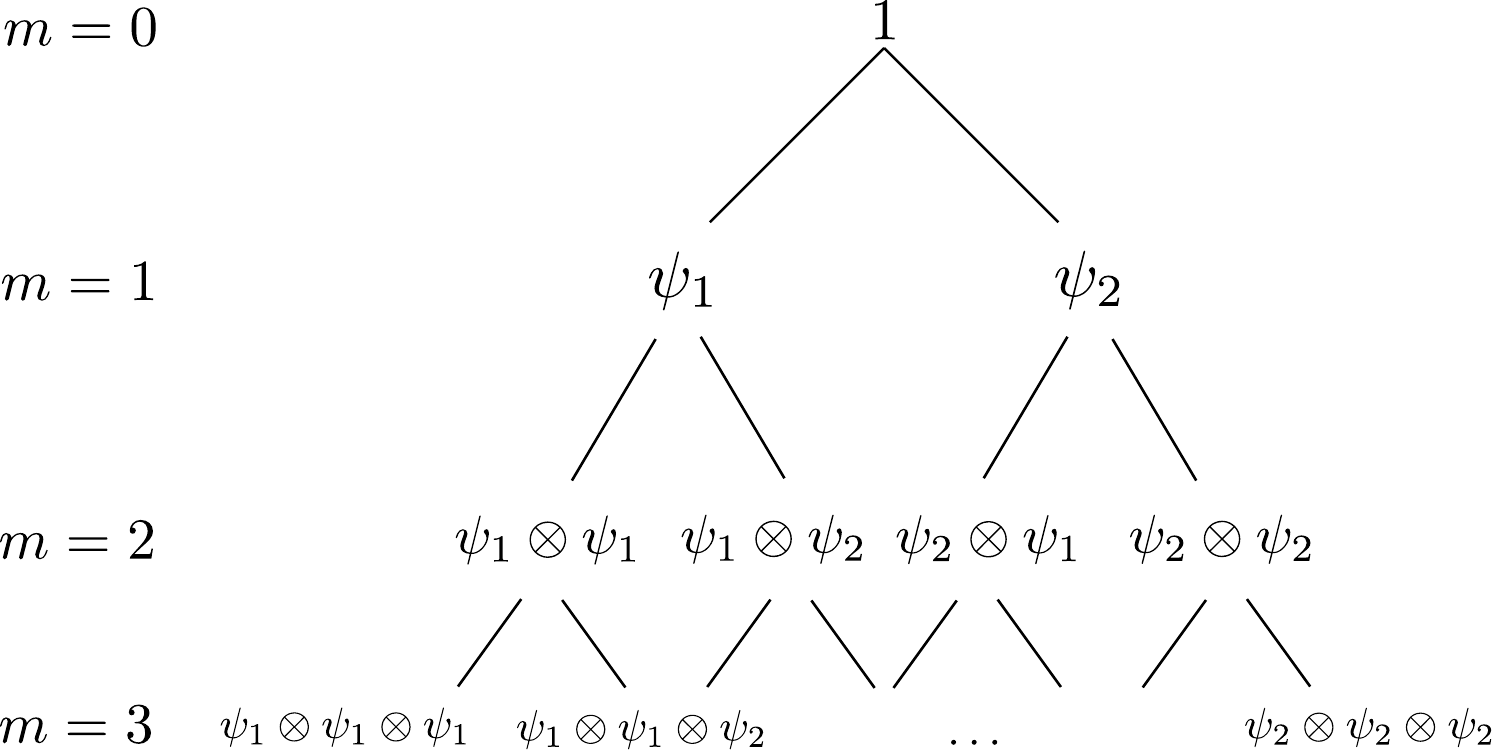}
\caption{Tensor-products tree of the eigenvectors of the evolution
matrices $M_m$ (with $m=0,1,2,3$) for a quadratic dimer, according
to \eqref{psi_m}. The index $m=0$ corresponds to the identity
element. This tree is a solution for the iterative polytope chains in
\figref{fig1}(a).}
    \label{fig2}
\end{figure}

\subsection{Iterated polytope chains in the field-moments space of
bosonic systems}

A chain of $m$ evolution matrices $M_1\to M_2\to\dots\to M_m$, in
\eqref{M_m},
 induces an IPC of $m$ $i_k$-polytopes,
 \begin{eqnarray}\label{Pc}
 P_{i_1}\to P_{i_2}\to\dots\to P_{i_m},
 \end{eqnarray}
produced by iterated Cartesian powers of $P_{i_1}$; namely,
\begin{equation*}
    P_{i_k}=\left(P_{i_1}\right)^k,
\end{equation*}
where $P_{i_1}$ is actually a
simplex, i.e., the simplest polytope in a given dimension
associated with the first-order FMs space. 

The dimension of a
polytope $P_{i_k}$ is $k\times d$, where $d$ is the dimension of
the $P_{i_1}$ polytope. The vertices of a polytope $P_{i_k}$ are
given by the elements of the FMs vector
${\langle\bigotimes^k_1\hat\Psi\rangle}$. The vertex potentials
and link weights are determined by the diagonal and off-diagonal
elements of the matrix $M_k$, respectively. Moreover, because the
matrix $M_k$ inherits the symmetry of $M_1$, according to
\eqref{M_m}, the links in $P_{i_k}$ preserve the symmetry of those
in the polytope $P_{i_1}$.   Note that, in general, the
polytope links can be asymmetric, i.e., with two different weights
in opposite directions, if the FMs evolution matrices are
non-symmetric.

For instance, in case of a symmetric bosonic dimer, an evolution
matrix $M_1$ can be associated with a 1D polytope $P_1$, i.e., a
line, formed  by two vertices and one edge [see \figref{fig1}(a)].
A 2-polytope $P_2$, corresponding to $M_2$, then represents a
square. As a result, one obtains an $m$-hypercube, describing  the
$m$th-order FMs space, whose associated matrix is $M_m$ [see
\figref{fig1}(a)]. Figure~\ref{fig1}(b) illustrates the first two
polytopes corresponding to the first and second-order FMs of a
symmetric three-mode system.

\subsection{Tensor product eigenspace for iterated polytope chains}

Here we discuss the eigenspace of the evolution matrices $M_m$
associated with  IPCs. Equation~(\ref{M_m}) implies that one can
determine a complete eigendecomposition of the matrix $M_m$,
governing $m$th-order FMs, knowing the eigenvalues and
eigenvectors of the matrix $M_1$, and, therefore, the eigenspace
of corresponding $i_m$-polytope (see   \figref{fig2} and
Appendix~\ref{AA} for details). Namely, the eigenvectors of the
matrix $M_m$ are found via all $(2N)^m$ combinations of the tensor
products of the eigenvectors $\psi_{j}^{(1)}$, $j=1,\dots,2N$ of
the matrix $M_1$, together with the {\it harmonic} eigenvalues:
\begin{equation}\label{psi_m}
\psi^{(m)}_{i_1,i_2,\dots,i_m}=\psi_{i_1}^{(1)}\otimes\dots\otimes\psi_{i_m}^{(1)}, \quad \lambda^{(m)}_{i_1,i_2,\dots,i_m}=\sum\limits_{k=1}^m\lambda_{i_k}^{(1)},
\end{equation}
where each index $i_k=1,\dots,2N$, for each $k=1,\dots,m$, and $\lambda^{(1)}_{i_k}$ are the eigenvalues of the matrix $M_1$ (see Appendix~\ref{AA}). 

The expression in \eqref{psi_m} implies that the FMs space,
associated with IPCs,  can exhibit a large extra algebraic (diabolic)
degeneracy, which is rapidly increasing with the FMs order.
{In all subsequent discussions, unless stated otherwise, we will call algebraic degeneracy the extra degeneracy linked with the construction of the IPCs. We will not refer to the other types of algebraic degeneracies, including those
which are intrinsically associated with the spectrum of the generating matrix $M_1$.}
The akgebraic degeneracy stems from the fact that the different combinations of the eigenvalues
$\lambda_{i_k}$ may result in the same sum in \eqref{psi_m},
but correspond to different eigenvectors. One can then
straightforwardly calculate the algebraic degeneracy $D$ for a
given eigenvalue $\lambda^{(m)}_{i_1,i_2,\dots,i_m}$, as  
\begin{equation}
  D(\lambda^{(m)}_{i_1,i_2,\dots,i_m})=\dfrac{m!}{{n_{i_1}!n_{i_2}!\dots
n_{i_m}!}}
  \label{Degeneracy}
\end{equation}
where $n_{i_k}$ denotes the number of times the index $i_k$
appears in \eqref{psi_m}.

Noticeably, this diabolic degeneracy, originating in the IPCs
eigenspace, can lead to a nontrivial interplay between diabolic points and
exceptional points~\cite{Ozdemir2019}, provided that an evolution matrix $M_1$
possesses the latter. This can result in the formation of hybrid
diabolically-degenerate exceptional points (see Appendix~\ref{AC} for details and also Ref.~\cite{Perina2022}).
And while it may be challenging to directly access and utilize the
properties of such hybrid degeneracies in the FMs space, they can
be engineered and probed with a relative ease in the spectrum of
real-space systems exhibiting unique dynamical features~\cite{arkhipov2023b}.

\begin{figure*}[t!]
    \centering
    \includegraphics[width=\textwidth]{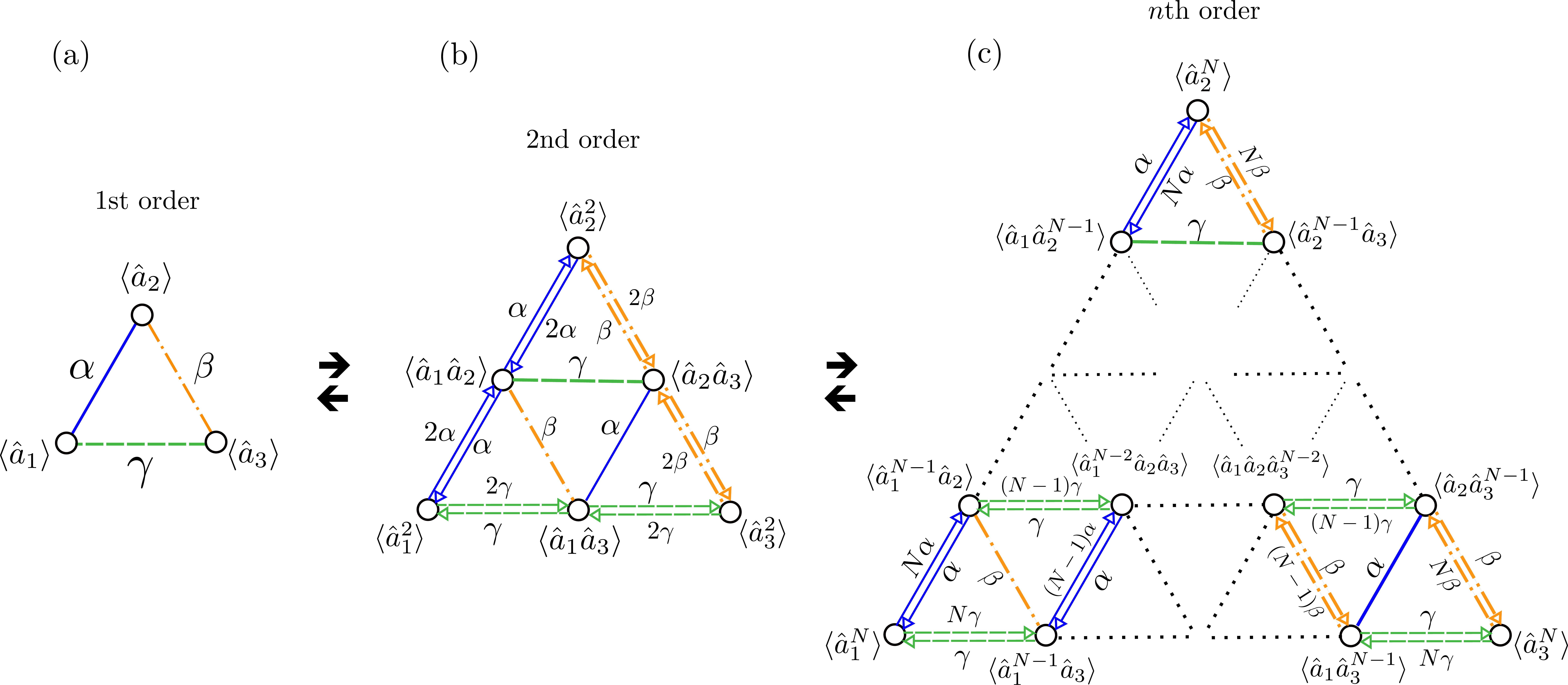}
\caption{Formation of iterative simplex chains (ISCs) by folding iterative polytope chains in the field-moments (FMs) space of
bosonic three-mode systems. (a) A 2-simplex, i.e., a triangle,
formed in the first-order FMs, which coincides with the $P_{i_1}$
polytope. The links, denoted by $\alpha,\beta,\gamma$ of different
colors, correspond to the coupling strengths between the three
FMs, according to \eqref{M1}.
(b) ISC in the second-order FMs space. By folding the
corresponding symmetric  $4$-polytope in \figref{fig1}(b), one
arrives at an asymmetric triangular 2-simplex consisting of six
vertices, according to \eqref{S}. (c) ISC in the $N$th-order FMs
space. The ISCs, thus, form a finite triangular non-Hermitian
lattice with $(N+1)(N+2)/2$ vertices in the high-order FMs space.
Formally, the presented 2-simplex chain can be considered in a
reversed order, that is, a lattice in panel (c) can be reduced
into a lattice in panel (a), which highlights the simple origin of
the seemingly complex lattice in (c). {Note that such formed asymmetric simplex lattices in the FMS of quadratic bosonic systems can be additionally symmetrized (see Appendix~\ref{AF} for details).}}
    \label{fig3}
\end{figure*}

\subsection{Dimensional reduction of iterated polytope chains to iterated simplex chains, and the formation of finite simplex lattices}

As mentioned above, the extra algebraic degenerate space in IPCs
stems from the tensor-product-states nature of the FM vectors
${\langle\bigotimes^m_1\hat\Psi\rangle}$. Indeed, in a simple case
of a linear dimer consisting of interacting fields, described by
$\hat a_1$ and $\hat a_2$, the system dynamics for the FMs
$\langle \hat a_1\hat a_2\rangle$ and $\langle \hat a_2\hat
a_1\rangle$ is the same,  though both these moments appear in the
FMs vector ${\langle\bigotimes^2\hat\Psi\rangle}$ by construction.
However, one can effectively eliminate this redundant degeneracy
by reducing $(2N)^m\times(2N)^m$-dimensional evolution matrices
$M_m$, describing $k_m$-polytopes, to low-dimensional effective
matrices $M_{m}^{\rm eff}$ of the size
\begin{eqnarray}\label{S}
    S_m(2N) =\dfrac{(2N+m-1)!}{m!(2N-1)!}<(2N)^m.
\end{eqnarray}
The expression for $S_m$ is the standard formula for {\it
combinations with repetition}, which leaves only non-equivalent
FMs in the description of the FMs time evolution.

{The projection of the matrix $M_m$ onto $M_{m}^{\rm eff}$ can be
generally represented as  
\begin{equation}
  M_m^{\rm eff}=TL^{\dagger} M_m R T^{-1},
  \label{projection}
\end{equation}
where $L$ and $R$ are $(2N)^m\times S_m(2N)$ rectangular matrices comprised by the left and right eigenvectors of the matrix $M_m$, respectively. They project the matrix $M_m$ onto its low
$S_m$-dimensional non-degenerate eigenspace, and where  $T$ is a $S_m\times S_m$
similarity transformation, e.g.,  a rotation, which allows to
write down $M_m^{\rm eff}$ in an arbitrary basis (see
Appendix~\ref{AB} for more details).} 
Since the eigenspace of the
matrix $M_m$ can be readily calculated in \eqref{psi_m}, one then
immediately obtains the eigenspace of the reduced matrix $M_m^{\rm
eff}$.

Interestingly, the series of effective evolution matrices
\begin{equation*}
    M_1^{\rm eff}\to M_2^{\rm eff}\to\dots\to M_m^{\rm eff}
\end{equation*}
can now
be associated with an ISC, in analogy with \eqref{Pc}. In other
words, the dimensional reduction of evolution matrices $M\to
M_{\rm eff}$ is accompanied  by the space reduction of the
corresponding IPC into ISC. Moreover, such ISCs are defined in the
same dimensional space as the simplex (polytope) $P_1$. 

Recall
that for the first-order FMs, the polytope $P_1$ is actually a
simplex. For instance, for a dimer, whose  1-simplex $P_1$ is a 1D
line segment [see \figref{fig1}(a)], its ISC defines a line
consisting of $[S_m(2)=m+1]$ vertices in the reduced $m$th-order
FMs space. Accordingly, for a trimer, with $P_1$ being a triangle,
i.e., a 2-simplex,  [\figref{fig1}(b)], the ISC defines a
triangular chain on 2D plane (see \figref{fig3}). Similarly, the
high-order FMs space of an $n$-mode system can be represented by
$(n-1)$-simplex chains.

This space reduction of IPCs into ISCs can formally be considered
as a projection of $k_m$-polytopes onto an $(n-1)$-dimensional
hypersurface which is formed by $n$ points corresponding to the
$m$th order of $n$ fields, i.e., $\langle\hat
a_1^m\rangle,\langle\hat a_2^m\rangle\dots,\langle\hat
a_n^m\rangle$.
After such a projection, a $k_m$-polytope is folded into a $(n-1)$-simplex lattice with $S_m(n)$ vertices. 
The vertex and link weights of the formed simplex lattice are then
assigned by the corresponding effective matrix $M_m^{\rm eff}$.
For instance, for a dimer, the 1D-hypersurface (a line) is the
main diagonal of the $(k_m=m)$-hypercube formed by two `vertices'
$\langle\hat a_1^m\rangle$ and $\langle\hat a_2^m\rangle$. As
such, $2^m$ vertices of a $m$-cube are projected to a
$S_m(2)=(m+1)$ vertices lying on the main diagonal with the vertex
and link weights described by the effective matrix $M_m^{\rm
eff}$.

 {Importantly, the reduction of  IPCs into ISCs results in the formation of finite $n$-simplex lattices which do {\it not} possess translational symmetry, regardless whether the underlying and generating matrix $M_1$ is symmetric or not (see also Sec.~\ref{Example}). This also implies that these constructed simplex structures can only describe {\it finite} lattices with {\it open} boundaries, because no periodic conditions can be imposed. The latter restriction evidently stems from the requirement that these lattices have a harmonic spectrum.  
}

{The finite lattices, emerging in the FMs space of quadratic
bosonic systems, can additionally emulate various real-space and other synthetic 
lattice models described by (non-) Hermitian Hamiltonians, whose matrix representation is equivalent to some $M_m^{\rm
eff}$. 
{\it Meaning that a high-order FMs space of
bosonic systems can be used to find exact solutions of various
finite (non-) Hermitian lattice models in arbitrary dimensions.} 
In other words, if the geometry of an arbitrary Hamiltonian, which can be
expressed via a n-simplex lattice (in a real or synthetic
space), can be identified with that of the constructed simplex structures studied here (see Sec.~\ref{Example} and Appendix~\ref{AD} for typical examples), one can then readily obtain its exact eigenspectrum. 
This observation constitutes one of the main results of our work. 
Note, however, that the state spaces and the physical meaning of the state vectors can substantially differ depending on the representation of the emulated Hamiltonian.
In Sec.~\ref{III}, we additionally discuss the practical relevance and other potential physical scenarios, where these $n$-simplex lattices  can occur and be effectively utilized.}

\subsection{Example: Exactly solvable finite triangular non-Hermitian
lattices with open boundaries}\label{Example}

To elaborate more on the above described procedure of obtaining
finite simplex lattices with open boundaries in the FMs space, as an example, let us
consider a linear three-mode system. Assume that its first-order
FMs vector  $\langle\hat\Phi\rangle$, $\hat\Phi=[\hat a_1, \hat
a_2, \hat a_3 ]^T$, is governed by the following symmetric
evolution matrix [see also \figref{fig3}(a)]:
\begin{eqnarray}\label{M1}
    M_1=\begin{pmatrix}
        0 & \alpha & \gamma \\
        \alpha & 0 & \beta \\
        \gamma & \beta & 0
    \end{pmatrix}.
\end{eqnarray}
The diagonal elements of $M_1$ can be in general complex-valued.
The evolution matrix $M_N$, governing the dynamics of the
$N$th-order FMs vector ${\langle\bigotimes^N_1\hat\Phi\rangle}$,
is then obtained from $M_1$ according to \eqref{M_m}. The $M_N$
describes a highly degenerate $2N$-polytope [the case for $N=1,2$
is shown in \figref{fig1}(b)]. This redundant degeneracy can be
eliminated via the map $M_N\to M_N^{\rm eff}$, which geometrically
corresponds to the reduction of a symmetric $2N$-polytope, with
$3^N$ vertices, to a 2D ISC forming a finite non-Hermitian
triangular lattice with 
\begin{equation*}
    S_N(3)=(N+1)(N+2)/2
\end{equation*}
 vertices (see
\figref{fig3}). The non-Hermitian nature of such a formed lattice
is highlighted by the arising asymmetry in vertex bonds which
become direction-dependent.

One can see in \figref{fig3} that by iteratively increasing the
FMs order, the corresponding ISC are iteratively augmented by
2-simplexes, (i.e., triangles), generating thus a triangular
lattice. This emerging triangular lattice can mimic a real-space
bosonic   non-Hermitian Hamiltonian $M_N^{\rm eff}\to
H_{\rm eff}^N$ written in a mode representation. We give an
explicit expression for this Hamiltonian in Appendix~\ref{AD}.
Importantly, this triangular lattice model in \figref{fig3}(c) is
exactly solvable, since the eigenspace of the matrix $M_N^{\rm
eff}$ is readily obtained from $M_N$ in \eqref{psi_m} by applying
the described projection.

Similarly, as was mentioned earlier, a $m$-dimensional hypercube,
related to the $m$th-order FMs of a dimer, can be reduced to a 1D
non-Hermitian chain [depicted as one of the external boundaries of
the triangular lattice in \figref{fig3}(c)] which is described by
a tridiagonal Sylvester matrix $M_k^{\rm
eff}$~\cite{arkhipov2023a}. 
Moreover,  for a four-mode system one
attains a tetrahedral-octahedral honeycomb lattice (see
Appendix~\ref{AD}). It is straightforward to extend this procedure
to arbitrary $n$-mode systems, whose FMs space can be described by
the corresponding $(n-1)$-simplex lattices with open boundaries.

Note that one can try to
symmetrize a non-Hermitian matrix $M_m^{\rm eff}$ by choosing an
appropriate basis if any. For example, for a 2-simplex lattice, shown in \figref{fig3},
one can find such a basis
where the corresponding reduced matrix $M_m^{\rm eff}$ becomes
symmetric (see also Appendix~\ref{AF} for details).

{This analysis can also be straightforwardly extended to any kind of a (non-) Hermitian matrix, which defines the evolution of the first-order field moments (away from exceptional points~\cite{KatoBOOK}). As the simplex lattices are characterized by the bidirectional couplings, the asymmetry in $M_1$ is simply reflected in the modification of the corresponding link weights. For instance, by modifying the single element $M_1(1,2)=\eta$ in \eqref{M1}, (i.e., changing only the tunneling amplitude from the mode $\hat a_1$ to the mode $\hat a_2$ ),   the weights of links, denoted by the blue-colored vectors at $45^\circ$ in panels (b) and (c) of \figref{fig3},  become factorized by the new parameter $\eta$. The remaining link weights remain unchanged.}
{
The symmetry of the underlying matrix $M_1$, which is explicitly reflected in its spectrum, then persists in any subsequent matrix $M_m^{\rm eff}$, and, therefore, in the associated $n$-simplex structures. The latter property follows from the harmonic character of the spectrum of  $M_m^{\rm eff}$, constructed by the eigenvalues of $M_1$. For instance, if the matrix $M_1$ possesses the chiral or parity-time ($\cal PT$) symmetry, that symmetry is also present in the formed $n$-simplex lattices. 
}

\section{Applications}\label{III}
\subsection{Discrete fractional Fourier transform}
Here we show that 1D simplexes provide a physical setting for the computation of the DFrFT in the FMs space of a photonic symmetric dimer.  
As it has been already demonstrated in Ref.~\cite{Weimann2016},  by mapping the matrix of the angular momentum operator $J_x$ onto a real-space Hamiltonian, describing the 1D array of coupled waveguides, one can effectively implement the DFrFT of any input signal. Our goal  here is to show that the effective matrix $M_N^{\rm eff}$, which determines the evolution of the $N$th-order FMs space of the symmetric dimer and is associated with a synthetic finite 1-simplex array, can be cast onto the $J_x$ matrix. 

Geometrically, this 1-simplex array can be represented as one of the three main outer sides of the 2D simplex in Fig.~\ref{fig3}(c), e.g., by the outer left-hand-side of the lattice, where the link weights factored by the parameter $\alpha$. This 1D-chain of the 2D simplex describes the FMs interaction between the modes $\hat a_1$ and $\hat a_2$, while eliminating the presence of the mode $\hat a_3$ [see \figref{fig3}(c)]. The elements of the $(N+1)\times(N+1)$ FMs evolution matrix $M_N^{\rm eff}$ then has the following Sylvester tridiagonal form~\cite{Hu2021,Arkhipov2021a} (see also Appendix~\ref{AC}):
\begin{eqnarray}
    M_{N}^{\rm eff}(j,k) =&& (N-j)\alpha \delta_{j, k-1} + j\alpha \delta_{j, k+1}, 
\end{eqnarray}
with indices $j,k=0,\dots,N$.
The matrix $M_N^{\rm eff}$ can be transformed into the operator of the quantum angular momentum (up to a scaling factor 2$\alpha$) via a similarity transformation:
\begin{eqnarray}\label{Jx}
    J_x=S^{-1}M^{\rm eff}_NS,
\end{eqnarray}
 where the diagonal matrix $S$ reads (see Appendix~\ref{AF})
 \begin{equation}\label{Sgen}
    S(j,j)=\prod\limits_{k=1}^j\sqrt{\frac{M_N^{\rm eff}(k,k-1)}{M_N^{\rm eff}(k-1,k)}},
\end{equation}
with $S(0,0)=1$. The matrix $J_x$ has the form
\begin{eqnarray}
    J_x(m,n)=&&{\alpha}\Big[\sqrt{(j-m)(j+m+1)}\delta_{m+1,n} \nonumber \\
    &&+\sqrt{(j+m)(j-m+1)}\delta_{m-1,n}\Big].
\end{eqnarray}
For simplicity, we keep the  $(N+1)\times(N+1)$ matrix $J_x$ in the standard form, where the indices $m$ and $n$ range from $-j$ to $j$ in unit steps, and $N=2j$. The parameter $j$ is a positive integer or half-integer.
Thus, the 1D simplex lattice, arising in the FMs space of a symmetric dimer, becomes {\it similar} to the operator space of the quantum angular momentum.

One of the outcomes, stemming from the above correspondence, is the ability to compute the DFrFT of real-space 1D signals in the {\it synthetic} FMs space of the dimer. This can be achieved, for instance, by appropriately preparing the initial quantum state of a dimer, whose input vector of the field moments can then emulate a given 1D real-space signal. 
Such a task can be achieved in quantum hardware through quantum optimal control algorithms, such as, e.g., GRAPE~\cite{Doria2011} or CRAB~\cite{KHANEJA2005} algorithms, 
Most importantly, since the FMs space are robust against any perturbations by construction, it makes the computation of DFrFT resilient to additional noise, which inevitably arises due to experimental imperfections when engineering 1D arrays in real space~\cite{Weimann2016}.
As the further step of our study, it would be  interesting to examine the potential of the revealed  harmonic $n$-simplex lattices to serve as a general framework for implementing the DFrFT in arbitrary dimensions for any $n>1$.

We note that the isomorphism between the angular momentum operators and bosonic creation and annihilation operators of two-mode systems is known in the literature by exploiting the Jordan-Schwinger transformations from which one can recover the same operator Lie algebra~\cite{Sakurai2020}. Here, however, and quite remarkably, we have unveiled the same angular momentum-bosonic dimer correspondence without any references to the system operators defined on the Fock Hilbert space, dealing {\it exclusively} with the geometry of $c$-number-valued FMs. Moreover, we extend this correspondence to arbitrary $m$-mode quadratic systems (see Appendix~\ref{AF} for details). This also means that experimental simulations of angular momenta can also be directly implemented in the FMs space of a dimer without relying, e.g., on postselection procedures, which directly exploit the Jordan-Schwinger transformations in the photon-number space of the system operators~\cite{Asano2019}, and which, therefore, are hard to realize in practice for many-photon systems.

\subsection{Other applications}
The possibility to engineer chiral and $\cal PT$-symmetric $n$-simplex structures also enable to study various topological, localization, and order-disorder transition  phenomena in high dimensions. Thanks to their harmonic spectrum, one can easily engineer high-dimensional simplex lattices with, e.g.,  exceptional points or even hybrid diabolically-degenerate exceptional points~\cite{arkhipov2023b} of an arbitrary degeneracy. It is well known that the presence of highly-degenerate exceptional points may induce  various non-Hermitian phase transitions~\cite{Bergholtz2021,Heiss1991}. In particular, in Ref.~\cite{arkhipov2023a} we have already unveiled the emergence of the non-Hermitian skin effect in the synthetic space of a $\cal PT$-symmetric dimer, whose FMs space corresponds to the 1-simplex array. The latter observation thus necessitates a further study of the associated non-Hermitian phenomena which can occur in high-dimensional $n$-simplex lattices for arbitrary large $n>1$.

We note that compared to Ref.~\cite{arkhipov2023a}, here we explicitly give the recipe for the construction of arbitrary dimensional $n$-simplex lattices in the FMs space, and, moreover, we found a general solution for their eigendecomposition. In particular, we present an alternative spectral solution for tridiagonal Sylvester matrices, whose exact eigenspectrum has been obtained only recently~\cite{Hu2021}. Therefore, we believe that our findings could also motivate and be of interest for future research in related mathematical and theoretical fields.

Besides the described possible theoretical and experimental applications of the given $n$-simplex lattices, primarily, they offer a simple visual representation and give valuable insights into the structure of the FMs space of bosonic quadratic systems. This result, to the best of our knowledge, was not previously known.
Additionally, the exact eigendecomposition of these simplex structures provides an alternative solution for the high-order FMs of quadratic systems, thus bypassing the need to invoke the Wick theorem~\cite{Wick1950}.

\section{Conclusions}

We have revealed the formation of nontrivial exactly solvable
non-Hermitian harmonic finite lattice models {with open boundaries} in the FMs space of arbitrary
quadratic bosonic $n$-mode systems. The spectrum of such lattices,
represented by ISCs, can be readily obtained by knowing the
spectrum of higher-dimensional matrices describing IPCs, which are
much easier to compute by exploiting the simple eigenspace
structure of the latter.

Our results extend and generalize the previous works in
Refs.~~\cite{Narducci1972,Hu2021,arkhipov2023a}, and provide
valuable insights into the underlying structure of many-body
lattice systems exhibiting similar complexity. These findings can
prompt further studies on simulation of various {\it
high-dimensional} lattice topological and localization phenomena
in the FMs of multimode  bosonic systems. The results can also be
directly utilized for finding high-order FMs of quadratic systems
even without invoking the Wick theorem.

{One of our main observations in this study is that some lattice structures, which look complex at first sight, can originate from simple low-dimensional arrangements. This insight, thus, resonates with the concept of quasicrystalline structures, which can be obtained through the projection of high-dimensional regular crystals onto certain hyperplanes~\cite{Jagannathan2021}, or can even draw parallels with fractal-like formations, where complex geometric patterns arise from much simpler geometric objects via iterative transformations~\cite{Falconer2014}.}

{These findings also enable a potential application of $n$-simplex lattices to provide a physical setting for realizing an arbitrary dimensional DFrFT either directly in the synthetic FMs space or in real space. This can pave the way for numerous intriguing applications in integrated quantum computation and various information processing tasks.}

The last but not least, the revealed here exotic hybrid spectral
features, namely, {\it emerging diabolically degenerate exceptional points
in the unfolded bosonic FMs space}, can be motivating for
constructing real-space Hamiltonians describing multimode
systems with unique dynamical
characteristics~\cite{arkhipov2023b}.

\acknowledgements

A.M. is supported by the Polish National Science Centre (NCN)
under the Maestro Grant No. DEC-2019/34/A/ST2/00081. S.K.O.
acknowledges support from Air Force Office of Scientific Research
(AFOSR) Multidisciplinary University Research Initiative (MURI)
Award on Programmable systems with non-Hermitian quantum dynamics
(Award No. FA9550-21-1-0202) and the AFOSR Award FA9550-18-1-0235.
F.N. is supported in part by: Nippon Telegraph and Telephone
Corporation (NTT) Research, the Japan Science and Technology
Agency (JST) [via the Quantum Leap Flagship Program (Q-LEAP), and
the Moonshot R\&D Grant Number JPMJMS2061], the Asian Office of
Aerospace Research and Development (AOARD) (via Grant No.
FA2386-20-1-4069), and the Foundational Questions Institute Fund
(FQXi) via Grant No. FQXi-IAF19-06.

\appendix

\section{Tensor product states in the higher-order field moments dynamics
of quadratic systems: Kronecker products and sums} \label{AA}

In this section we  detail on the construction of the evolution
matrices for higher-order  FMs based on the Kronecker sum
operations and their eigendecomposition.

As we have already stressed, one of the remarkable features of
quadratic systems is that one can obtain the analytical form of an
evolution matrix ruling the dynamics of any higher-order FMs from
the form of the evolution matrix $M_1$ for the first-order FMs in
\eqref{M_m}. This can be done by exploiting some properties of
matrices formed by Kronecker sums. First, we summarize the most
important spectral features of Kronecker sum matrices in the
following theorem.
\begin{theorem}
{\it Theorem}. The eigenspectrum of a complex matrix $C\in{\mathbb
C}^{(m+n)\times(m+n)}$,  obtained as the Kronecker sum of two
complex matrices $A\in{\mathbb{C}^{m\times m}}$ and
$B\in{\mathbb{C}^{n\times n}}$, i.e., $C = A\oplus B=A\otimes
I_B+I_A\otimes B$, has the form
\begin{eqnarray}\label{eigval}
\lambda(C) = \lambda(A)+\lambda(B), \quad \psi_C=\psi_A\otimes
\psi_B,
\end{eqnarray}
where $\lambda(M)$ and $\psi_M$ are the eigenvalues and
eigenvectors of a matrix $M=A,B,C$, respectively.
 \end{theorem}
 \begin{proof}
The proof is straightforward. We start from the
eigenvalue-eigenvector equation for the matrix $C$, by feeding it
into the equation for the right eigenvector, which is a tensor
product of the two eigenvectors $\psi_A$ and $\psi_B$, with
eigenvalues $\lambda(A)$ and $\lambda(B)$, respectively. Namely,
\begin{eqnarray}\label{eigvec}
C(\psi_A\otimes&&\psi_B)=(A\oplus B)\nonumber  \\
&&= (A\otimes I_B)( \psi_A\otimes \psi_B)+(I_A\otimes B)( \psi_A\otimes  \psi_B)\nonumber \\
&&=(A \psi_A\otimes I_B \psi_B)+(I_A \psi_A\otimes B \psi_B)\nonumber \\
&&=(\lambda_A \psi_A\otimes \psi_B)+( \psi_A\otimes \lambda_B \psi_B) \nonumber \\
&&=(\lambda_A+\lambda_B) \psi_A\otimes \psi_B,
\end{eqnarray}
where we used the tensor and dot product properties of matrices
and vectors. In other words, the eigenvector of the matrix
$C=A\oplus B$, corresponding to the eigenvalue $\lambda(C)$, is
indeed just the tensor product of the two eigenvectors of the
matrices $A$ and $B$ with eigenvalues $\lambda(A)$ and
$\lambda(B)$.
\end{proof}

According to \eqref{eigvec}, the matrix, whose eigenvectors are
formed by the tensor products of eigenvectors of  two matrices can
be utilized in the construction of the evolution matrices of any
higher order. To reveal how this works in practice, let us
consider the second-order field moments. The various combinations
of second-order field moments are obtained from the tensor product
of the Nambu vector,
 \begin{eqnarray}
     \hat\Psi=\left[\hat a_1,\hat a_2,\dots,\hat a_N,\hat a_1^{\dagger},\hat a_2^{\dagger},\dots,\hat a_N^{\dagger}\right]^T,
 \end{eqnarray}
 on itself, i.e., from the $4N^2$ dimensional vector $\langle\hat \Psi\otimes\Psi\rangle$.
 According to \eqref{eigvec}, the evolution matrix $M_2$, governing the vector of the second-order FMs,
 \begin{equation}
     \frac{{\rm d}}{{\rm d}t}\langle\hat \Psi\otimes\Psi\rangle=M_2\langle\hat \Psi\otimes\hat\Psi\rangle,
 \end{equation}
 attains the form
 \begin{equation}\label{M2}
     M_2=M_1\oplus M_1 = M_1\otimes I_{2N}+I_{2N}\otimes M_1.
 \end{equation}
The form of the evolution matrix $M_2$ coincides, as it should,
with that derived from the Lyapunov equation for the covariance
matrix, when the latter is presented as a column
vector~\cite{Lototsky2015}. This procedure can, thus, be
iteratively continued to any $m$th order field moment vectors
${\Big\langle\bigotimes\limits_1^m\Psi\Big\rangle}$, thus
obtaining \eqref{M_m}.

The eigenvectors space for such evolution matrices $M_m$ governing
$m$th-order FMs can be, thus, represented as a tree structure,
according to \eqref{psi_m}. For instance, for a dimer, this tree
structure takes the form shown in \figref{fig3}.

\section{Dimensional reduction of polytopes to simplexes in the
field-moments space of bosonic systems}\label{AB}

As we have already stressed, the highly-degenerate eigenspace of a
$(2N)^m\times (2N)^m$ matrix $M_m$ can be effectively reduced to
the non-degenerate space of  a matrix $M_m^{\rm eff}$ of size
$S_m(2N)\times S_m(2N)$. This space folding can be associated with
the dimensional reduction of a polytope to a simplex. The
reduction of the degenerate space can be performed in accordance
with the following algorithm, which allows to simultaneously
retrieve both the reduced matrix $M_m^{\rm eff}$ and its
eigenvectors.
\begin{itemize}
\item Define $(2N)^m\times (2N)^m$ matrices $A$ and $B$, which consist
of all the eigenvectors of $M_m$ and basis vectors
$v_k=[0,\dots,1_k,\dots,0]^T$, respectively:
$A=[\psi_1,\dots,\psi_m]$, and $B=[v_1,\dots,v_m]$. The matrix
$B$, thus, can be chosen as the identity matrix.
\item Eliminate all extra degenerate eigenvectors in $A$ corresponding
to a $D_i$-fold eigenvalue $\lambda_i$, such that among all the
degenerate eigenmodes $\psi^{(N)}_{i_1,i_2,\dots,i_N}$  only one
representative eigenvector 
\begin{eqnarray}\label{Di}
    \psi^{m}(\lambda_i)=\dfrac{1}{D_i}\sum\limits^{D_i}\psi^{(N)}_{i_1,i_2,\dots,i_N}
\end{eqnarray}
is left in the matrix. This transforms  $A$ to a matrix $A'$ of
the size $(2N)^m\times S_m$, which determines the non-degenerate
eigenspace of $M_m$ one wants to project on.

\item Repeat the same procedure for the matrix $B\to B'$, collecting and
eliminating the columns with the same indices as in the map $A\to
A'$. Matrix $B'$ defines the basis for the reduced matrix
$M_m^{\rm eff}$.

\item {Project the $(2N)^m\times(2N)^m$ matrix $M_m$ onto its non-degenerate $S_m(2N)$-dimensional eigenspace spanned by the columns in $A'$, which results in the formation of a diagonal $S_m(2N)\times S_m(2N)$ matrix $D$. This can be done either by simply filling the diagonals of $D$ with the corresponding eigenvalues $\lambda$ of the column-eigenvectors in $A'$, or more formally via the action $D=A''^{\dagger}M_m A'$, where $A''$ is dual to $A'$, with the property $A''^{\dagger}A'=I$, and which is comprised by the corresponding left eigenvectors of $M_m$. }


\item The matrix $M_m^{\rm eff}$ is then found as $M_m^{\rm
eff}=TDT^{-1}$, where the matrix  $T={B'}^T A'$, whose columns are
comprised by the eigenvectors of $M_{m}^{\rm eff}$.
\end{itemize}
Following these steps allows one to straightforwardly find the
reduced effective matrix $M_{m}^{\rm eff}$ and its eigenvectors
from the original high-dimensional matrix $M_m$.

\section{Example: Quantum parametric subharmonic generation processes}\label{AC}

To explicitly illustrate the existence of IPCs with the tensor-product-states
structure and diabolically-degenerate exceptional points (DDEPs) in the FMs space of quadratic systems, we now
consider the FM dynamics (up to third order) for parametric
subharmonic generation processes. Previous
works~\cite{Wang2019,Roy2020,Flynn2020} have already revealed the
existence of exceptional points (EPs) in such non-dissipative quadratic Hamiltonians.
Here, we show that such non-Hermitian systems can possess
additional nontrivial spectral features in the FMs space, i.e.,
DDEPs.

Let us first start from a quadratic Hermitian Hamiltonian $\hat H$
describing a second subharmonic generation with a classical pump,
and working in the reference frame rotating at the pump frequency
$\omega_p$:
 \begin{eqnarray}
      \hat H = \Delta\hat a^{\dagger}\hat a+({ig}/{2})\left(\hat a^{2}-\hat a^{\dagger 2}\right),
 \end{eqnarray}
where $\Delta = \omega_p-\omega$ is the resonance detuning, i.e.,
the difference between the frequencies of the pump ($\omega_p$)
and quantum field ($\omega$), the parameter $g$ is assumed to be a
real-valued coupling constant, which involves the amplitude of the
pump field~\cite{Perina1991Book}.

\subsection{First-order field-moments dynamics}

The dynamics of the first-order moments of the Nambu operator
vector $\hat \Psi=\left[\hat a,\hat a^{\dagger}\right]^T$ obeys
\eqref{M_m} with the evolution matrix
\begin{eqnarray}
      M_1 = \begin{pmatrix}
-i\Delta & -g \\
-g & i\Delta
\end{pmatrix}.
\end{eqnarray}
The matrix $M_1$ is $\cal
PT$-symmetric~\cite{ChristodoulidesBook,El-Ganainy2018}, i.e., it
is invariant under action ${\cal PT}M_1({\cal PT})^{-1}$, of the
parity $\cal P$ and time-reversal ($\cal T$) operators, which
action is defined as $
    {\cal PT}=\hat\sigma_x{\cal K},
$ where the operator $\cal K$ accounts for the complex conjugate
operation. The symmetric matrix $M_1$ describes a complex
1-polytope, shown in \figref{fig1}(a), with two vertices having
complex weights $\pm i\Delta$ and a single edge with weight $-g$.
The eigenvalues of $M_1$ are $\lambda_{1,2}=\pm\Lambda$, where
$\Lambda=\sqrt{g^2-\Delta^2}$. The corresponding first-order
moment eigenvectors become
\begin{equation}\label{ss_psi}
    \psi_{1,2} \equiv \begin{pmatrix}
    \mp\exp(\mp i\phi) \\
    1
    \end{pmatrix}, \quad \phi=\arctan({\Delta}/{\Lambda}),
\end{equation}
which satisfy biorthogonality. Both the eigenvalues and
eigenvectors coalesce at the second-order EP, defined by the
condition: $
    g_{\rm EP}=\Delta,
$ with the phase $\phi_{\rm EP}=\pi(2k+1)/2$, $k\in\mathbb{Z}$, in
\eqref{ss_psi}.

\subsection{Second-order field-moments dynamics}\label{ACII}

The evolution matrix $M_2$, governing the second-order FM vector
\begin{eqnarray}
    \langle \hat\Psi\otimes\hat\Psi\rangle = \left[\langle\hat a^2\rangle,\langle\hat a\hat a^{\dagger}\rangle,\langle\hat a^{\dagger}\hat a\rangle, \langle \hat a^{\dagger 2}\rangle\right]^T,
\end{eqnarray}
and describing a 2-polytope in the from of a square [see
\figref{fig1}(a)] reads
\begin{equation}\label{M2ex}
    M_2 = \begin{pmatrix}
    -2i\Delta & -g & -g & 0 \\
    -g & 0 & 0 & -g \\
    -g & 0 & 0 & -g \\
    0 & -g & -g & 2i\Delta
    \end{pmatrix}.
\end{equation}
The eigenvalues of $M_2$ are $\lambda^{(2)}_{jk}=2\Lambda{\rm
diag}[1,-1]$, where all four eigenvalues are listed in the matrix
form, with $j,k=1,2$.

The eigenvectors of  $M_2$ are then obtained via all $2^2=4$
combinations of tensor products of the eigenvectors $\psi_{1,2}$
in \eqref{ss_psi}, according to \eqref{psi_m}:
\begin{eqnarray}
    \psi^{(2)}_{i_1,i_2}=\psi_{i_1}\otimes\psi_{i_2},
\end{eqnarray}
where each index $i_k=1,2$, for each $k=1,2$. Namely,
\begin{eqnarray}\label{A}
    \psi^{(2)}=\begin{pmatrix}
        e^{-2i\phi} & -1 & -1 & e^{2i\phi} \\
        -e^{-i\phi} & -e^{-i\phi} & e^{i\phi} & e^{i\phi} \\
        -e^{-i\phi} & e^{i\phi} & -e^{-i\phi} & e^{i\phi} \\
        1 & 1 & 1 & 1
    \end{pmatrix}.
\end{eqnarray}
The matrix $M_2$ possesses only a single non-degenerate
exceptional point (NDEP) of third-order, i.e., there is no DDEP in
the second-order FM space.

Below we comment on the reduction of the degenerate dimensions in
the IPC to the ISC in the second-order FMs space of the  system
considered. The two-fold degenerate eigenvalue
$\lambda_{12}^{(2)}=\lambda_{21}^{(2)}=0$ arises because of the
redundant moment elements $\langle\hat a\hat a^{\dagger}\rangle$
and $\langle\hat a^{\dagger}\hat a\rangle$ in the moments vector
$\langle \hat\Psi\otimes\hat\Psi\rangle$. As such, the
$(2^2=4)$-dimensional eigenspace of the second-order moments can
be decreased to a $[S_2(2)=3]$-dimensional one.
\begin{itemize}
    \item The matrix $A$ has the same form as in \eqref{A}. The matrix $B$ is just an identity matrix, i.e., $B=I_4$.
    \item Since only two eigenvectors  [the second and third columns in
    \eqref{A}]: 
    \begin{eqnarray}
        \psi_2&=&[-1,-e^{-i\phi},e^{i\phi},1]^T, \nonumber \\
        \psi_3&=&[-1,e^{i\phi},-e^{-i\phi},1]^T,
    \end{eqnarray}
    are degenerate, belonging to the zero-valued eigenvalue, we
    construct a new eigenvector $\psi=(\psi_2+\psi_3)/2$. Thus, we
    obtain a new $4\times 3$ matrix $A'$, 
    \begin{eqnarray}\label{A'}
        A'=\begin{pmatrix}
            e^{-2i\phi} & -1 & e^{2i\phi} \\
            -e^{-i\phi} & i\sin\phi & e^{i\phi} \\
            -e^{-i\phi} & i\sin\phi & e^{i\phi} \\
            1 & 1 & 1
        \end{pmatrix}.
    \end{eqnarray}
    \item Similarly one attains the matrix $B'$ as follows 
    \begin{eqnarray}\label{B'}
        B'=\begin{pmatrix}
            1 & 0 & 0 \\
            0 & {1}/{2} & 0 \\
            0 & {1}/{2} & 0 \\
            0 & 0 & 1
        \end{pmatrix}.
    \end{eqnarray}
    \item The $3\times 3$ diagonal matrix $D$ then reads as 
    \begin{eqnarray}\label{D}
        D={\rm diag}[\Lambda,0,-\Lambda].
    \end{eqnarray}
    \item By constructing the matrix $T={B'}^TA'$ in the form 
    \begin{eqnarray}
        T = \begin{pmatrix}
    e^{-2i\phi} & -1 & e^{2i\phi} \\
    -e^{-i\phi} & i\sin(\phi) & e^{i\phi} \\
    1 & 1 & 1
    \end{pmatrix},
    \end{eqnarray}
    one readily acquires the reduced effective matrix  $M_{2}^{\rm
    eff}=TDT^{-1}$, which reads 
    \begin{eqnarray}\label{M2eff}
        M_{2}^{\rm eff}=\begin{pmatrix}
    -2i\Delta & -2g & 0\\
    -g & 0 & -g \\
    0 & -2g & 2i\Delta
    \end{pmatrix}.
    \end{eqnarray}
\end{itemize}
The eigenvectors of the matrix $M_2^{\rm eff}$ are respectively
the columns of the matrix $T$.

\subsection{Diabolic degeneracy in the eigenspace of the third-order field
moments: Diabolically Degenerate Exceptional Points}

Let us now elaborate in detail on the appearance of a DDEP in the
third-order FMs space. The corresponding $8\times8$ evolution
matrix $M_3$, describing the dynamics of the vector of moments
$\Big\langle \bigotimes\limits_1^3\hat\Psi\Big\rangle$ 
defined as
\begin{eqnarray}
    \Big\langle \bigotimes\limits_1^3\hat\Psi\Big\rangle =\Big[&&\langle\hat a^3\rangle,\langle\hat a^2\hat a^{\dagger}\rangle,\langle\hat a\hat a^{\dagger}\hat a\rangle,\langle\hat a\hat a^{\dagger 2}\rangle, \langle \hat a^{\dagger}\hat a^2\rangle,\langle \hat a^{\dagger}\hat a\hat a^{\dagger}\rangle, \nonumber \\
    &&\langle \hat a^{\dagger 2}\hat a\rangle,\langle \hat a^{\dagger 3}\rangle\Big]^T.
\end{eqnarray}
according to \eqref{M_m} reads
\begin{equation}
    M_3 = \left( \begin {array}{cccccccc} -3\,i\Delta&-g&-g&0&-g&0&0&0
\\ \noalign{\medskip}-g&-i\Delta&0&-g&0&-g&0&0\\ \noalign{\medskip}-g&0
&-i\Delta&-g&0&0&-g&0\\ \noalign{\medskip}0&-g&-g&i\Delta&0&0&0&-g
\\ \noalign{\medskip}-g&0&0&0&-i\Delta&-g&-g&0\\ \noalign{\medskip}0&-
g&0&0&-g&i\Delta&0&-g\\ \noalign{\medskip}0&0&-g&0&-g&0&i\Delta&-g
\\ \noalign{\medskip}0&0&0&-g&0&-g&-g&3\,i\Delta\end {array} \right).
\end{equation}
Its eigenvalues $\lambda_{ijk}$, according to \eqref{psi_m}, in
the main text, can be listed as
\begin{eqnarray}
  \lambda^{(3)} = \begin{pmatrix}
  \lambda_{111} & \lambda_{112} & \lambda_{121} &\lambda_{122} &\lambda_{211} &\lambda_{212} &\lambda_{221} &\lambda_{222} \\
  3\Lambda & \Lambda & \Lambda & -\Lambda & \Lambda & -\Lambda & -\Lambda & -3\Lambda
  \end{pmatrix}, \nonumber \\
\end{eqnarray}
where as before $\Lambda = \sqrt{g^2-\Delta^2}$.
\begin{figure}[!htb]
    \centering
    \includegraphics[width=0.49\textwidth]{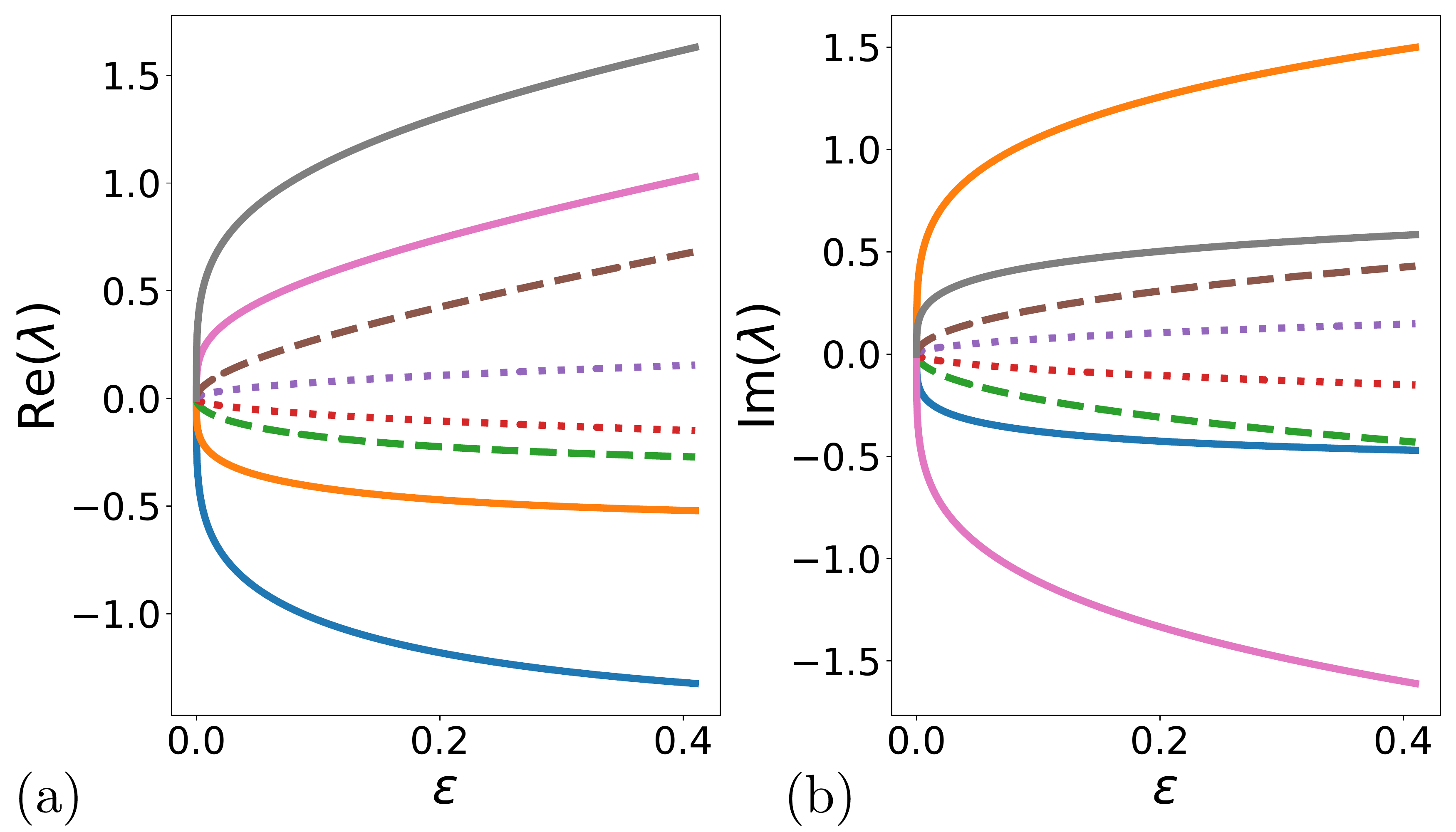}
\caption{Real (a) and imaginary (b) parts of the eigenvalues
$\lambda$ of the matrix $M_3$, governing third-order FMs, at the
EP $g=\Delta$, under the perturbation $\epsilon$, which enables
resolving a DDEP. The matrix $M_3$ has a single EP of fourth order
(presented by the coalescence of four solid curves, characterized
by the quartic-root dependence on perturbation,
$\sqrt[4]{\epsilon}$). Apart from it, under the perturbation,  the
initial DDEP splits into two EPs of second order (presented by the
merge of two eigenvalues, shown by the two dashed and two dotted
curves, respectively, which are characterized by the square-root
dependence  of the perturbation,  $\sqrt{\epsilon}$).}
    \label{fig4}
\end{figure}
The corresponding eigenvectors are easily found, according to
\eqref{psi_m},
\begin{widetext}
\begin{eqnarray}
  \psi^{(3)}=\left( \begin {array}{cccccccc} -{{\rm e}^{-3\,i\phi}}&{{\rm e}^{-i
\phi}}&{{\rm e}^{-i\phi}}&-{{\rm e}^{i\phi}}&{{\rm e}^{-i\phi}}&-{
{\rm e}^{i\phi}}&-{{\rm e}^{i\phi}}&{{\rm e}^{3\,i\phi}}
\\ \noalign{\medskip}{{\rm e}^{-2\,i\phi}}&{{\rm e}^{-2\,i\phi}}&-1&-1
&-1&-1&{{\rm e}^{2\,i\phi}}&{{\rm e}^{2\,i\phi}}\\
\noalign{\medskip}{ {\rm e}^{-2\,i\phi}}&-1&{{\rm
e}^{-2\,i\phi}}&-1&-1&{{\rm e}^{2\,i\phi }}&-1&{{\rm
e}^{2\,i\phi}}\\ \noalign{\medskip}-{{\rm e}^{-i\phi}}&-{ {\rm
e}^{-i\phi}}&-{{\rm e}^{-i\phi}}&-{{\rm e}^{-i\phi}}&{{\rm e}^{i
\phi}}&{{\rm e}^{i\phi}}&{{\rm e}^{i\phi}}&{{\rm e}^{i\phi}}
\\ \noalign{\medskip}{{\rm e}^{-2\,i\phi}}&-1&-1&{{\rm e}^{2\,i\phi}}&
{{\rm e}^{-2\,i\phi}}&-1&-1&{{\rm e}^{2\,i\phi}}\\
\noalign{\medskip}- {{\rm e}^{-i\phi}}&-{{\rm e}^{-i\phi}}&{{\rm
e}^{i\phi}}&{{\rm e}^{i \phi}}&-{{\rm e}^{-i\phi}}&-{{\rm
e}^{-i\phi}}&{{\rm e}^{i\phi}}&{ {\rm e}^{i\phi}}\\
\noalign{\medskip}-{{\rm e}^{-i\phi}}&{{\rm e}^{i \phi}}&-{{\rm
e}^{-i\phi}}&{{\rm e}^{i\phi}}&-{{\rm e}^{-i\phi}}&{ {\rm
e}^{i\phi}}&-{{\rm e}^{-i\phi}}&{{\rm e}^{i\phi}}
\\ \noalign{\medskip}1&1&1&1&1&1&1&1\end {array} \right).
\end{eqnarray}
\end{widetext}
Note that the two eigenvalues $\pm\Lambda$ are triply degenerate.
At the EP $g=\Delta$, the Jordan form of $M_3$ reads
\begin{eqnarray}
  M_3^{\rm EP} = \left( \begin {array}{cccccccc} 0&1&0&0&0&0&0&0\\ \noalign{\medskip}0
&0&1&0&0&0&0&0\\ \noalign{\medskip}0&0&0&1&0&0&0&0
\\ \noalign{\medskip}0&0&0&0&0&0&0&0\\ \noalign{\medskip}0&0&0&0&0&1&0
&0\\ \noalign{\medskip}0&0&0&0&0&0&0&0\\
\noalign{\medskip}0&0&0&0&0&0 &0&1\\
\noalign{\medskip}0&0&0&0&0&0&0&0\end {array} \right),
\end{eqnarray}

The effective dimension of $M_3$ is four [$S_3(2)=4$], and,
therefore, there is a single EP of fourth order. However, apart
from this EP,  also one DDEP of second order emerges, due to the
triply degenerated pairs of eigenvalues.

In order to lift the degeneracy of the DDEP, one can induce a
specific perturbation to the matrix $M_3$, i.e., $M_3\to
M_3+\epsilon P$, where $\epsilon$ denotes the perturbation
strength.  The matrix $P$ can read:
\begin{equation}
    P = {\rm diag}\left[1,0,0,1,0,1,0,0\right].
\end{equation}
The result of such a perturbation on the eigenvalues of $M_3$ is
shown in \figref{fig4}. Other choices of perturbation do not
necessarily lead to the DDEP detection. Note that the perturbation
is actually fictitious because of the symmetry protection of the
evolution matrix $M_3$ in the FM space. Indeed, by unfolding the
evolution matrices $M_3$ from the first-order moments $2\times2$
matrix $M_1$, it is impossible to attain such a perturbed matrix
$M_3$. However, the formation of such DDEP in the spectrum can be
checked by mapping the evolution matrix $M_3$ to a certain
real-space photonic-lattice Hamiltonian.

Importantly, the above revealed hybrid spectral features of FMs
polytopes can also be exploited in the construction of similar
non-Hermitian Hamiltonians in real space. For instance, one of the nontrivial outcomes
of the presence of a DDEP in the spectrum of real multimode non-Hermitian Hamiltonians
can be the implementation of a programmable multimode switch by
dynamically encircling the DDEP~\cite{arkhipov2023b}. A direct
detection of DDEPs in the FMs space of a bosonic system can be
performed by measuring bosonic commutators~\cite{Zavatta2009}
and/or anomalous FMs~\cite{Kuhn2017}. The tensor-product-states structure of a FM
space can also be recovered from Cahill-Glauber $s$-ordered FMs,
which can be obtained from the measured normally-ordered moments
using standard photon-detection
schemes~\cite{Perina1991Book,Perina2017,Shchukin2005}. Note that
our conclusions remain valid also for the quadrature field
moments, which could be easier to  access experimentally, e.g.,
via standard balanced homodyne
detection~\cite{Wang2019,Roccati2021b}.

\section{Triangular and Tetrahedral  non-Hermitian lattice models in the
field-moments space of three- and four-mode systems,
respectively}\label{AD}

\subsection{Triangular non-Hermitian lattice model with open boundaries}

\begin{figure}[!htb]
    \centering
    \includegraphics[width=\columnwidth]{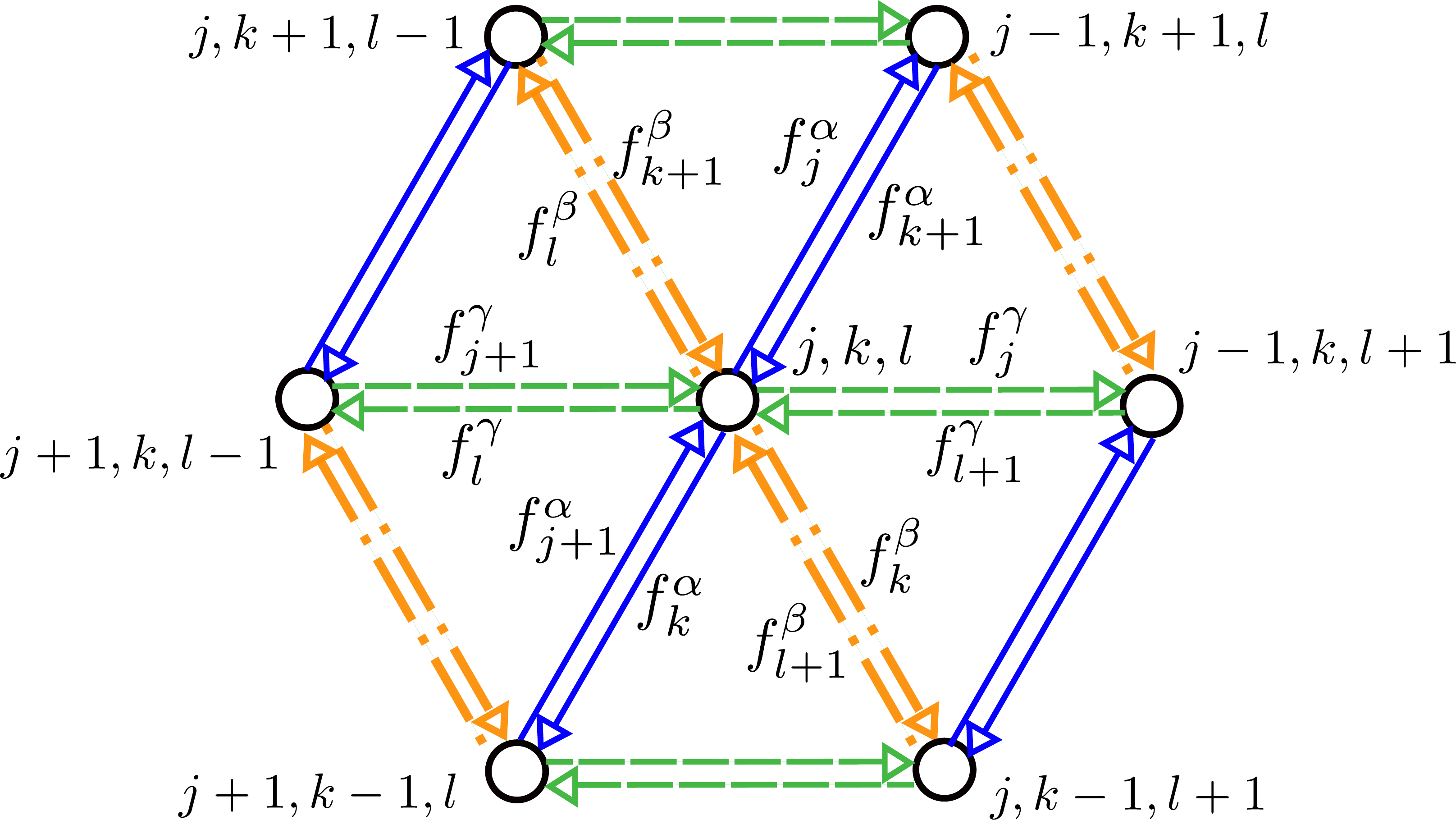}
\caption{An element of the triangle-shaped triangular
non-Hermitian lattice model in \figref{fig3}(c), described by the
non-Hermitian Hamiltonian in \eqref{H}, illustrating the nearest-neighbor interactions
for a site with indices $(j,k,l)$ (shown in the centre). The site
potentials and couplings $f^{\xi}_q$ are given in Eqs.~(\ref{H0})
and (\ref{Hint}), respectively.}
    \label{fig5}
\end{figure}

 We have already mentioned about the emergence of a triangular
lattice in the FMs space of a three-mode system, as depicted in
\figref{fig3}. Here we discuss in a more detail the structure of
such a lattice, its mapping to a non-Hermitian Hamiltonian, and
shortly mention the tetrahedral lattices appearing in the FMs
space of four-mode systems.

For the sake of clarity, we repeat some parts of the main text.
First, let us consider the first-order FMs vector
$\langle\hat\Phi\rangle$ of a three-mode system, where
$\hat\Phi=[\hat a_1, \hat a_2, \hat a_3 ]^T$, which is governed by
the following symmetric evolution matrix: 
\begin{eqnarray}
     M_1=\begin{pmatrix}
        \Delta_1 & \alpha & \gamma \\
        \alpha & \Delta_2 & \beta \\
        \gamma & \beta & \Delta_3
    \end{pmatrix}.
\end{eqnarray}
Compared to the \eqref{M1}, we also add the diagonal terms $\Delta_k$, which can be associated with different resonant (complex) frequencies of the interacting modes.
The degenerate evolution matrix $M_N$ governing the dynamics of
the $N$th-order FM ${\langle\bigotimes^N_1\hat\Phi\rangle}$ is
obtained from $M_1$ according to \eqref{M_m}. The matrix $M_N$
then describes a highly degenerate $2N$-polytope [the case for
$N=1,2$ is shown in \figref{fig1}(b)]. This extra degeneracy can
be eliminated by effectively reducing the $2N$-polytope with $3^N$
vertices to a two-dimensional triangular ISCs having $S_N(3)$
vertices instead (see also \figref{fig3}). This folding procedure
is accompanied by the reduction of the symmetric evolution matrix
$M_N$ to the asymmetric effective matrix $M_N^{\rm eff}$.

One then can relate the effective evolution matrix $M^{\rm eff}_N$
to a  non-Hermitian Hamiltonian describing a  triangle-shaped triangular lattice
(shown in \figref{fig3}). One can see that the hopping amplitudes
have the maximal values for the sites on the boundary, being
proportional to the size of the corresponding external or internal
triangle lattice edges. The tunneling amplitudes are steadily
decreasing when approaching the opposite sides of the triangular
lattice, revealing the characteristic asymmetry in the site
interactions, i.e., their path dependence. The eigendecomposition
of the non-Hermitian Hamiltonian $\hat H^N$, which is identical to the matrix $M^{\rm
eff}_N$, directly follows from the eigendecomposition of the
matrix $M_N$, as was discussed in the previous section. Since the
spectrum of $M_N$ can be readily calculated, so  the eigenspectrum
of the non-Hermitian Hamiltonian $\hat H^N$.

\begin{figure*}[t!]
    \centering
    \includegraphics[width=0.8\textwidth]{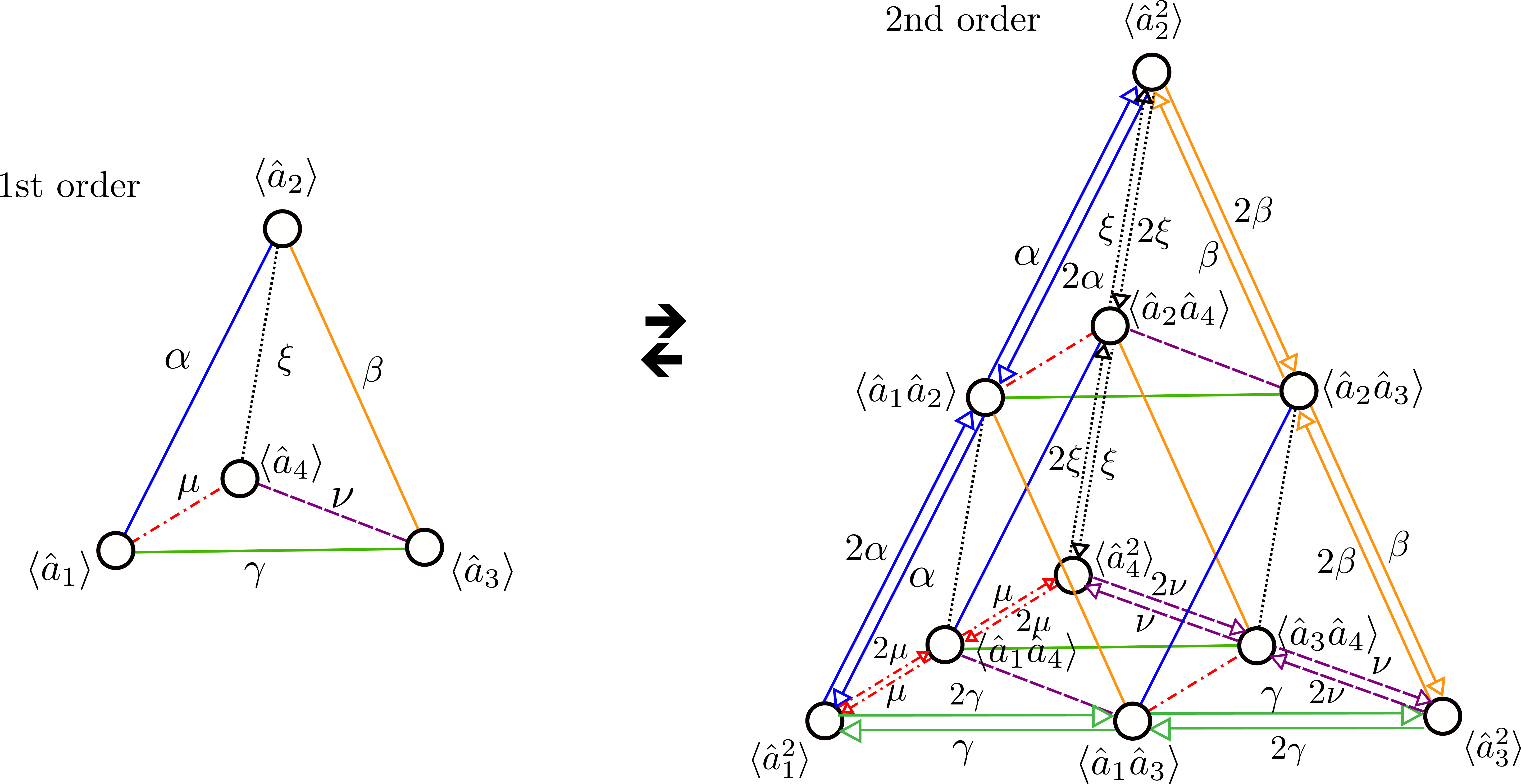}
    \caption{Formation of iterated simplex chains by folding iterated polytope chains in the field-moments space of bosonic four-mode systems. (a) A 3-simplex, i.e., a tetrahedron, formed in the first-order FMs, which coincides with the $P_{i_1}$ polytope. The links, denoted by Greek letters $\alpha,\beta,\gamma,\mu,\xi,\nu$, and highlighted by different colors, correspond to the coupling strengths between the three FMs. (b) An ISC in the second-order FMs space. By folding the corresponding symmetric  $6$-polytope, one obtains an asymmetric $3$-simplex, namely a tetrahedral-octahedral honeycomb in the high-order FMs space of the four-mode bosonic system.}
    \label{fig6}
\end{figure*}

Namely, by inducing the map $M^{\rm eff}_N\to\hat H^N$, one
attains:
\begin{eqnarray}\label{H}
    \hat H^N = \hat H_0^N + \hat H_{\rm int}^N,
\end{eqnarray}
where $\hat H_0^N $ accounts for the non-interacting Hamiltonian
part  describing `site potentials' on the triangular lattice,
which reads
\begin{eqnarray}\label{H0}
\hat H^N_0 = \sum\limits_{\substack{j,k=1\\l=N-j-k}}^N
\left(j\Delta_1+k\Delta_2+l\Delta_3\right)\hat
c^{\dagger}_{j,k,l}\hat c_{j,k,l},
\end{eqnarray}
and asymmetric `site interactions' are described by the part:
\begin{eqnarray}\label{Hint}
\hat H^N_{\rm int}=&&\sum\limits_{\substack{j,k=1\\l=N-j-k}}^N \hat c_{j,k,l}\Big(f_{j}^{\alpha}\hat c^{\dagger}_{j-1,k+1,l}+f_{j}^{\gamma}\hat c^{\dagger}_{j-1,k,l+1}\nonumber \\
&&+f_{k}^{\beta}\hat c^{\dagger}_{j,k-1,l+1}
+f_{k}^{\alpha}\hat c^{\dagger}_{j+1,k-1,l}\nonumber \\
&&+f_{l}^{\gamma}\hat c^{\dagger}_{j+1,k,l-1}+f_{l}^{\beta}\hat c^{\dagger}_{j,k+1,l-1}\Big) \nonumber \\
&&+\hat c^{\dagger}_{j,k,l}\Big(f_{k+1}^{\alpha}\hat c_{j-1,k+1,l}+f_{l+1}^{\gamma}\hat c_{j-1,k,l+1}\nonumber \\
&&+f_{l+1}^{\beta}\hat c_{j,k-1,l+1}
+f_{j+1}^{\alpha}\hat c_{j+1,k-1,l}\nonumber \\
&&+f_{j+1}^{\gamma}\hat c_{j+1,k,l-1}+f_{k+1}^{\beta}\hat
c_{j,k+1,l-1}\Big),
\end{eqnarray}
where the hopping coefficients read $f^{\xi}_{q}=q\xi$. An element
of such a lattice is also shown in \figref{fig5}.

\subsection{Tetrahedral non-Hermitian lattice model with open boundaries}\label{AC2}

In a similar fashion, one can obtain a $3$-simplex non-Hermitian
lattice, namely a non-Hermitian tetrahedral-octahedral honeycomb
model, emerging in the ISCc FMs space of a four-mode system, whose
first-order FMs symmetric evolution matrix has the following
generic form 
\begin{eqnarray}
    M_1=\begin{pmatrix}
        \Delta_1 & \alpha & \gamma & \mu \\
        \alpha & \Delta_2 & \beta & \xi \\
        \gamma & \beta & \Delta_3 & \nu \\
        \mu & \xi & \nu & \Delta_4
    \end{pmatrix}.
\end{eqnarray}
The schematic representation of the 3-simplex described by the
matrix $M_1$ is shown in \figref{fig6}(a). The second-order FMs
evolution matrix $M_2$ is then an augmented tetrahedron from below
[see \figref{fig6}(b)]. Similarly, one obtains an iterated
3-simplex chain comprised by tetrahedrons for arbitrary high-order
FMs.

\section{Symmetrization of non-Hermitian $n$-simplex lattices}\label{AF}
\begin{figure*}[t!]
    \centering
    \includegraphics[width=\textwidth]{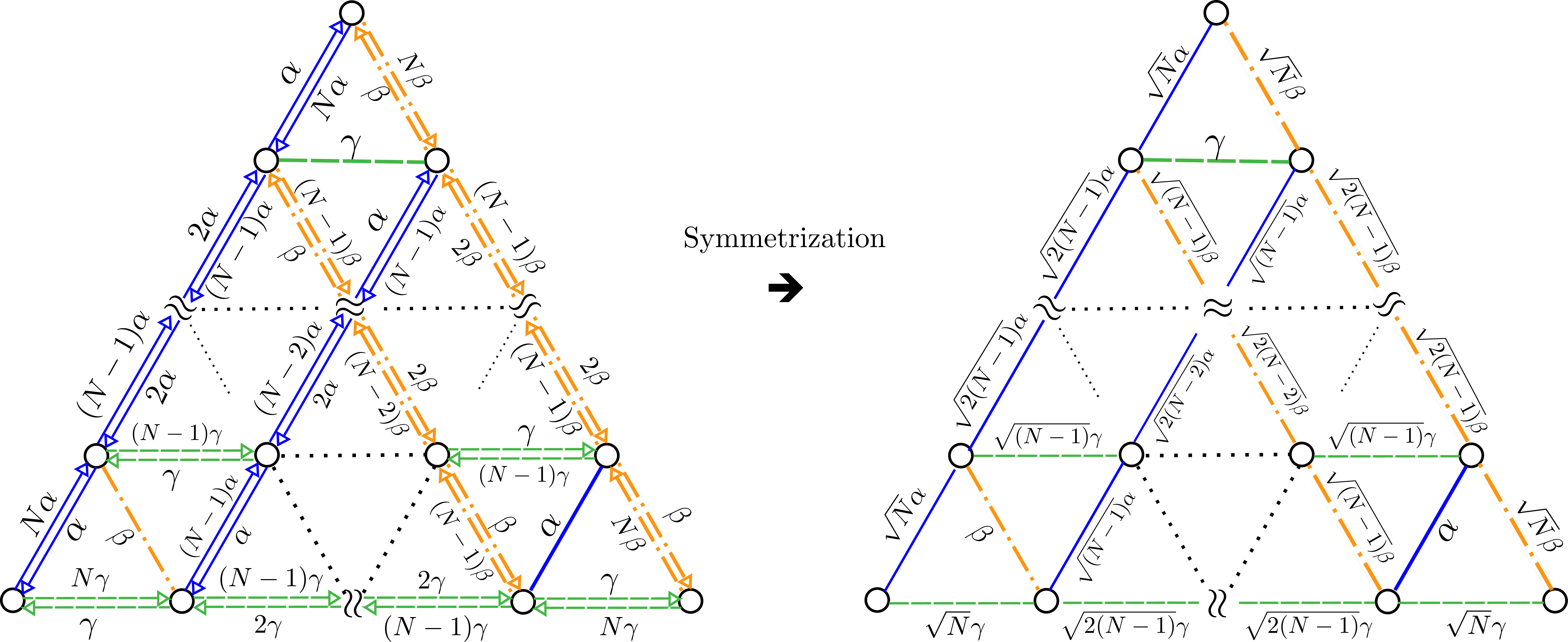}
    \caption{Symmetrization of the asymmetric 2-simplex lattice (left panel), also shown in \figref{fig3}(c), to the symmetric 2-simplex lattice (right panel). 
    In other words, the asymmetric intersite couplings (depicted as double arrows on the left panel) become symmetric (shown by single arrows on the right panel).
    This mapping can realized via an appropriate similarity transformation, or, alternatively, by choosing an appropriate basis when projecting polytopes onto simplexes, formed in the field-moments space of quadratic systems. The edges of the symmetric 2-simplex lattice can describe a quantum angular momentum operator with different quantum numbers. See more details in Appendix~\ref{AF}.}
    \label{fig7}
\end{figure*}
As was mentioned in Sec.~\ref{Example} and Sec.~\ref{III} of the main text, the asymmetric $n$-simplex lattices, emerging in the FMS of quadratic systems, can be transformed into symmetric ones by finding an appropriate similarity transformation, or, in other words, by choosing an appropriate basis when performing a projection of IPCs onto ISCs. This means that because of such lattice symmetrization only eigenvectors are modified, whereas eigenvalues are left the same.

In the third bullet of Appendix~\ref{AB}, we have pointed that a general $(2N)^m\times S_m(2N)$ matrix $B'$, which defines the basis of the polytopes projection on their non-degenerate space is composed of $S_m$ column-vectors, where $S_m$ is given in \eqref{S}. Some of these column-vectors are formed as a symmetrical combination, i.e., an average, of column-vectors of the identity matrix $I_{(2N)^m}$. 
However, if instead of that `average' of $D_i$ basis vectors, factored by $1/D_i$ (see \eqref{Di}), to take the same combination but factored by $1\sqrt{D_i}$, i.e., a normalized sum, then the new resulting transformation matrix $T=B'^TA'$ would produce a symmetric effective matrix $M_m^{\rm eff}$, which describes the symmetrized  $n$-simplex lattice.    

Let us explain this in more detail by considering a simple example, similar to that in Appendix~\ref{ACII}. That is, we would like to symmetrize the matrix $M^{\rm eff}_2$ in \eqref{M2eff}. For that, instead of the matrix $B'$, in \eqref{B'}, we take a new matrix $B''$, now formed by {\it normalized} column-vectors, i.e.,
\begin{equation}
    B'\to B''=\begin{pmatrix}
        1 & 0 & 0\\
        0 & \dfrac{1}{\sqrt{2}} & 0 \\
        0 & \dfrac{1}{\sqrt{2}} & 0 \\
        0 & 0 & 1
    \end{pmatrix}.
\end{equation}
By constructing a modified transformation matrix $T'=B''^TA'$, where $A'$ is the same as in \eqref{A'}, one obtains the symmetrical matrix $M_{2,\rm sym}^{\rm eff}$ in the form
\begin{eqnarray}
    M_{2,\rm sym}^{\rm eff}=T'DT'^{-1}=\begin{pmatrix}
        -2i\Delta & -\sqrt{2}g & 0 \\
        -\sqrt{2}g & 0 & -\sqrt{2}g \\
        0 & -\sqrt{2}g & 2i\Delta
    \end{pmatrix}, \nonumber \\
\end{eqnarray}
where diagonal matrix $D$ is given in \eqref{D}.
For a two-mode system, one can then show that an arbitrary $M_m^{\rm eff}$ matrix is similar to the symmetric matrix $M_{m,\rm sym}^{\rm eff}$ via the transformation in \eqref{Sgen}.

One can straightforwardly verify that this symmetrization procedure is extended to any effective matrix $M_m^{\rm eff}$ describing an $n$-simplex lattice. For instance, for a 2-simplex lattice, with asymmetrical bidirectional links as shown in \figref{fig3}(c), the described symmetrization procedure allows to transform that lattice into symmetrical one, where the intersite couplings become symmetric (see \figref{fig7}). 

Remarkably, due to this symmetrization of $M_m^{\rm eff}$, the resulting symmetric $n$-simplex lattices have 1D edges which can be mapped to the angular momentum operator $\hat J_z$ in the Fock space representation (up to a certain scaling factor, proportional to a bosonic mode coupling), provided that the diagonal elements of $M_m^{\rm eff}$ are zero. Indeed, according to \figref{fig7}, all (anti-) diagonal and horizontal edges of the symmetric 2-simplex lattice can represent the angular momentum operator with a certain azimuthal quantum number $j$, which is determined by the number of vertices belonging to the corresponding edge.

For instance, the submatrices of  $M_m^{\rm eff}$ describing three outer edges in \figref{fig7}, which contain $(N+1)$ vertices, correspond to the angular momentum operator with the quantum number $j=N/2$ (see also Sec.~\ref{III}), with additional scaling factors $\alpha$, $\beta$, and $\gamma$ (see \figref{fig7}). In general, an edge having $K$ vertices is related to the angular momentum quantum number $j=(K-1)/2$ (see \figref{fig7}).


%

\end{document}